\documentclass{llncs}
\usepackage[disable]{todonotes}
\usepackage{upgreek}
\usepackage{graphicx}
\usepackage{amsmath}
\usepackage{enumerate}
\usepackage[inline]{enumitem}
\usepackage{amssymb}
\usepackage{xspace}

\newcommand{\curve}[1]{\ensuremath{\mathbf{#1}}}
\newcommand{\region}[2]{\ensuremath{S_{\mathbf{#1}\mathbf{#2}}}}
\newcommand{\calR}{\ensuremath{\mathcal{R}}}
\renewcommand{\L}{\lefthalfcup}
\newcommand{\lefthalfcup}{\rule{0.4pt}{1.2ex}\rule{1.2ex}{0.4pt}\xspace}
\newcommand{\lefthalfcap}{\rule{0.4pt}{1.2ex}\rule[1.2ex]{1.2ex}{0.4pt}\xspace}
\newcommand{\righthalfcup}{\rule{1.2ex}{0.4pt}\rule{0.4pt}{1.2ex}\xspace}
\newcommand{\righthalfcap}{\rule[1.2ex]{1.2ex}{0.4pt}\rule{0.4pt}{1.2ex}\xspace}

\title{Order-preserving 1-string representations of planar graphs}
\titlerunning{Order-preserving 1-string representations} 
\author{Therese Biedl\thanks{Research supported by NSERC.} \and Martin Derka\thanks{The author was supported by the Vanier CGS.}}
\institute{David R. Cheriton School of Computer Science, University of Waterloo
\email{\{biedl,mderka\}@uwaterloo.ca}}

\begin{document}

\maketitle

\begin{abstract}
This paper considers 1-string representations of planar graphs that are 
{\em order-preserving} in the sense that the order of crossings along the
curve representing vertex $v$ is the same as the order of edges in the
clockwise order around $v$ in the planar embedding.  We show that this does
not exist for all planar graphs (not even for all planar 3-trees), but show
existence for some subclasses of planar partial 3-trees.  In particular, for
outer-planar graphs it can be order-preserving and {\em outer-string}
in the sense that all ends of strings are on the outside of the representation.
\end{abstract}

\section{Introduction}
\label{sec:introduction}

String representations recently received a lot of attention, especially
for planar graphs.  Scheinerman \cite{cit:scheinerman} had asked in 1984 whether every 
planar graph can be represented as the intersection graph of segments in the plane.  
This was settled partially by Chalopin, Gon\c{c}alves and Ochem \cite{cit:chalopin-string}, who showed
that every planar graph has a 1-string representation, i.e., a representation as an intersection
graph of strings such that any two strings may cross at most once.  Extending their result,
in 2009 Chalopin and Gon\c{c}alves
finally settled Scheinerman's conjecture in the positive~\cite{cit:chalopin-seg}. 
We later showed that 1-string representations of planar graphs
can be achieved even with orthogonal curves
with at most 2 bends \cite{cit:jocg}.  A number of other papers gave 
string representations for subclasses of planar graphs that are simpler to build
and/or have other useful properties, see for example
\cite{FMP91,KobourovUeckertVerbeek,cit:chaplick,cit:mfcs,cit:cccg}.  
Testing whether a graph has a string representation is NP-hard 
\cite{cit:kratochvil-II,cit:middendorf} and in NP \cite{cit:schaefer}; the latter is not obvious because string representations may require
exponentially many bends for non-planar graphs \cite{KratochvilM1994}.

\medskip\noindent{\bf Our results: }
In this paper, we study the following question: Does every planar graph
have a 1-string representation where the order of crossings along curves 
{\em preserves} the planar embedding in the sense that the order of crossings 
along the curve of $v$ corresponds to the cyclic order of edges around $v$
in some planar embedding?  This is motivated by that we found 
string representations quite hard to read; during our work on \cite{cit:jocg}
we struggled to verify correctness in some cases because the crossing of
curves for an edge occurred at unexpected places.  Furthermore, having an
order-preserving string representation could make it easier to create such
representations  by using the typical incremental approach that adds one 
vertex on the outer-face at a time; for this
it would be especially helpful if such representations were also 
{\em outer-string} in the sense that ends of strings are on the
infinite region defined by the representation.
We show the following:
\begin{itemize}
\item Not all planar graphs have order-preserving 1-string representations.  In fact,
	we can construct a planar 3-tree that has no such representation.
\item For some subclasses of planar partial 3-trees, 
	we construct order-pre\-ser\-ving 1-string
	representations.  For outer-planar graphs, these are 
	additionally outer-string (and use segments),
	while for the other graph classes we show that 
	order-preserving outer-1-string
	representations do not always exist.  
\end{itemize}


We are not aware of any previous results on order-preserving 1-string 
repre\-sen\-tations. (On the other hand, string-representations of planar
graphs obtained from contact representations are usually order-preserving,
but strings then intersect twice, at least for some edges.)
\todo{TB: Changed here to avoid some of the criticisms.}
The closest related results are on the {\em abstract graph realizability problem}
\cite{cit:kratochvil-II,cit:middendorf}, which asks to draw a graph such that only a
given set of edge-pairs are allowed to cross.  

\section{Definitions}

A {\em string representation} $\calR$ assigns a curve $\curve{v}$ in the plane to every vertex
$v$ in a graph in such a way that $(v,w)$ is an edge if and only if $\curve{v}$
intersects $\curve{w}$.  (Throughout the paper, bold-face $\curve{x}$ always
denotes the curve assigned to vertex $x$.)
We demand that $\curve{u}$ and $\curve{v}$ intersect 
only if there is a proper crossing,
i.e., any sufficiently small circle centered at an intersection-point 
crosses $\curve{u}$, $\curve{v}$, $\curve{u}$, $\curve{v}$ in that order.
(In particular no curve $\curve{u}$ should end on another curve $\curve{v}$, 
though such a touching-point could always be resolved into a proper crossing by
extending $\curve{u}$ a bit.)  We also do not allow three curves to share a
point.
A {\em 1-string representation} is a string representation such that
any two curves cross at most once. A {\em segment representation} uses straight-line
segments in place of strings.  A {\em $B_k$-VPG-representation} uses orthogonal curves with at most $k$ bends as strings.

A string representation $\calR$ divides the plane into connected regions.  The {\em contour}
is the infinite region of $\Bbb{R}^2-\calR$.  A string representation is called {\em weakly outer-string} if
all vertex curves are incident to the contour.  It is called {\em 
outer-string} if all vertex curves have an end incident to the contour.%
\footnote{One could distinguish this further by whether both ends must be
on the contour or whether one end suffices.
All our outer-string constructions have both ends on the contour, while
all our impossibility-results hold even if only one end is required to be
on the contour, so the distinction does not matter for the results in
our paper.}
A weakly outer-string 
representation can be made outer-string by ``doubling back'' along the curve
of each vertex, but this does not work for an outer-1-string representation,
because doubling back along the curve would make some curves cross twice.  
See \cite{Cabello2016,cit:Keil2016} and the references therein for more on outer-string representations.

In this paper, we only consider connected graphs. A graph is called {\em planar} if it can be drawn in the plane without crossing.
Such a planar drawing $\Gamma$ defines, by enumerating edges around vertices in clockwise order,
a {\em rotation scheme}, i.e., an assignment of a cyclic order of edges  at each
vertex.   From the rotation scheme, one can read the {\em faces}, i.e., the 
vertices and edges that are incident to each connected piece of $\Bbb{R}^2-\Gamma$.
A {\em plane graph} is a planar graph with a fixed rotation scheme. 
An {\em outer-planar graph} is a planar graph that has a rotation scheme such
that all vertices are incident to one face.  
An {\em outer-plane graph} is a plane graph with the rotation system that describes such an embedding.
A {\em $k$-tree} (used here only for $k=2,3$) is a graph that has a vertex order
$v_1,\dots,v_n$ such that $v_1,\dots,v_k$ is a clique, and each $v_i$ for
$i>k$ has exactly $k$ neighbours in $v_1,\dots,v_{i-1}$, and they form a clique.
A {\em partial $k$-tree} is a subgraph of a $k$-tree.  
Every outer-planar graph is a partial 2-tree.

Fix a rotation scheme of a graph. 
We say that a 1-string
representation is {\em order-preserving} with respect to the rotation
scheme if for any vertex $v$, we can walk along curve $\curve{v}$ from
one end to the other and encounter the crossings  with $\curve{w_1},\dots,
\curve{w_k}$ in the same order in which the neighbours $w_1,\dots,w_k$ of
$v$ appear
in the cyclic order of edges around $v$.  This leaves open the choice
which neighbour of $v$ should be $w_1$, since the order at $v$ is cyclic  while
the order along $\curve{v}$ is not.%
\footnote{Once we fix how to break up the cyclic order at all vertices, an 
order-preserving 1-string representation
can be described abstractly as a graph $H$ and can be realized if and only if
$H$ is planar.  Hence the problem is interesting only if we keep this choice.}

\section{Graphs with no order-preserving representations}
\label{sec:does-not-exits}

In this section, we show that there exist planar graphs that have no
1-string representation that preserves the order of any planar embedding.  
To define them, we need the following graph operation:  Given a plane
graph $G$, the {\em stellation} of $G$
is obtained by inserting a new vertex into every face of $G$, and
making it adjacent to all vertices incident to that face.  The 
{\em triple-stellation} of $G$ is obtained by stellating $G$ to 
get $G'$, stellating $G'$ to get $G''$, and finally stellating $G''$. 

\begin{lemma}
\label{lem:stellation}
Let $G$ be a plane graph with minimum degree 3 and 
at least $|V(G)|+1$ faces that are triangles.  Then the 
triple-stellation $G'''$ of $G$ has no order-preserving 1-string representation with respect to this rotation scheme.
\end{lemma}
\begin{proof}
Assume for contradiction we had such a 1-string representation $\calR$,
and let $\calR_G$ be the induced 1-string representation of $G$, which
is also order-preserving.  The following notation will be helpful:
If $a,c$ are neighbours of $b$,
then let $\curve{b}[a,c]$ be the stretch
of $\curve{b}$ between the intersection with $\curve{a}$ and 
$\curve{c}$. 

Consider a face-vertex-incidence in $G$, which can be described
by giving a vertex $b$ and two neighbours $a,c$ of $b$ that
are consecutive in the clockwise order at $b$.  
We call such a face-vertex-incidence {\em unbroken} if 
(in $\calR_G$)
$\curve{b}[a,c]$ contains no other crossing, else
we call it {\em broken}.  Since $\calR_G$ is order-preserving,
for every vertex $b$ in $G$ only one face-vertex-incidence at $b$
is broken.
Since $G$ has at least $|V(G)|+1$
triangular faces, there exists a face $T=\{u,v,w\}$ of $G$ such that 
all face-vertex-incidences at $T$ are unbroken.  We will find a contradiction
at the stellation vertices that were placed in $T$.  
See also Fig.~\ref{fig:stellation}.

Let $x$ be the vertex that (during the stellation of $G$ to get $G'$) was placed
in face $T$.  We claim that $\curve{x}$ must intersect $\curve{u}$ in
$\curve{u}[v,w]$.  To see this, recall that $\deg_G(u)\geq 3$, hence $u$ has at 
least one other neighbour $u'$ in $G$.  Since the face-incidence at $u$ is
unbroken, $\curve{u}[v,w]$  contains no other crossing of $\calR_G$, so $\curve{u'}$
intersects $\curve{u}$ outside this stretch.  Since $T$ is a face in $G$,
the (clockwise or counter-clockwise)
order of neighbours at $u$ in $G'$ contains $u',v,x,w$.  To maintain this order
in the string representation, the intersection between
$\curve{x}$ and $\curve{u}$ (in $\calR$) must be on $\curve{u}[v,w]$.   
Similarly one argues that
$\curve{x}$ intersects $\curve{v}[u,w]$ and $\curve{w}[u,v]$.

Let $C$ be the region bounded by
$\curve{u}[v,w]\cup \curve{w}[u,v]\cup \curve{v}[w,u]$.
Curve $\curve{x}$ intersects $\delta C$ three times, and no more since
curves intersect at most once in a 1-string representation. So 
$\curve{x}$ starts (say)
inside $C$,
crosses $\delta C$ to go outside, crosses $\delta C$ to go inside, 
and then crosses $\delta C$ again to end outside.   Between the second
and third crossing, $\curve{x}$ contains
a stretch that is inside $C$; after possible
renaming of $\{u,v,w\}$ we assume that this is $\curve{x}[v,w]$.  This stretch
splits 
$C$ into two parts, say $C'$ (incident to parts of $\curve{u}$)
and $C^r$ (incident to the crossing of $\curve{v}$ and $\curve{w}$).
\todo{Unimportant change-request:  
Make $\curve{w}$ and $\curve{v}$ cross ``the other way''.
Extend $\curve{y}$ to cross $\curve{w}$ from outside.}

\begin{figure}[ht]
\hspace*{\fill}
\includegraphics[width=0.2\linewidth]{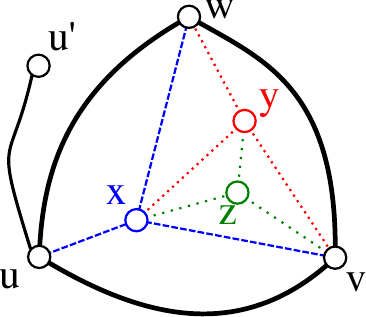}
\hspace*{\fill}
\includegraphics[width=0.35\linewidth]{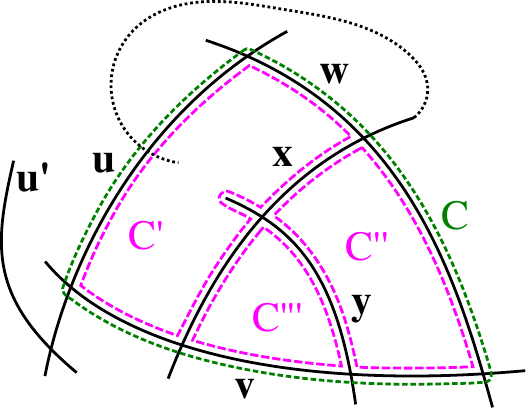}
\hspace*{\fill}
\caption{For the proof of Lemma~\ref{lem:stellation}.}
\label{fig:stellation}
\end{figure}

Let $y$ be the vertex that (during the stellation of $G'$ to get $G''$) was
placed in the face $\{v,w,x\}$ of $G'$.  Since $v,w,x$ all have degree
3 or more in $G'$, as before one argues that $\curve{y}$ must intersect
$\curve{x}[v,w]$, $\curve{w}[x,v]$ and $\curve{v}[w,x]$.
Curve $\curve{y}$ intersects $\delta C'$ (in $\curve{x}[v,w]$), but cannot 
intersect $\delta C'$ a second time, else it would cross $\curve{u}$ 
(but $(u,y)\not\in E$) or would
cross one of $\curve{x},\curve{v},\curve{w}$ twice (which is not allowed).
Hence $\curve{y}$ starts inside $C'$, then crosses $\curve{x}$, and then
crosses one of $\curve{v}$ and $\curve{w}$.  Up to renaming of $\{v,w\}$
we may assume that $\curve{y}$ crosses $\curve{v}$ first.
Hence $\curve{y}[x,v]$ splits $C^r$ into two parts, say $C''$ 
(incident to parts of $\curve{w}$)
and $C'''$ 
(incident to the crossing of $\curve{v}$ and $\curve{x}$).

Now finally consider the vertex $z$ that was placed in $\{x,y,v\}$
when stellating $G''$ to obtain $G'''$.  As before one argues that
$\curve{z}$ has an end inside $C'$, because it crosses $\curve{x}$
in stretch $\curve{x}[v,y]\subset \curve{x}[v,w]$, and it cannot cross
$C'$ again.  But we can also see that
$\curve{z}$ has an end inside $C''$, since it crosses $\curve{y}[x,v]$
and crosses no other curve on the boundary of $C''$.
But this means that $\curve{z}$ has
both ends outside $C'''$, contradicting that it must intersect the
boundary of $C'''$ three times to respect the edge-orders at $x,y,v$.
Contradiction, so $G'''$ does not have an order-preserving 1-string
representation. 
\qed
\end{proof}


\begin{theorem} 
\label{thm:notOrder}
There exists a planar 3-tree that has no order-preserving 
1-string representation.
\end{theorem}
\begin{proof}
Start with an arbitrary planar 3-tree $G$ with $n\geq 6$ vertices;
this has minimum degree 3 and $2n-4\geq n+2$ triangular faces   in its
(unique) rotation scheme.  Stellating a 3-tree gives again a 3-tree, so
by Lemma~\ref{lem:stellation} the triple-stellation of $G$ is a 3-tree that
has no order-preserving 1-string representation. \qed
\end{proof}

\section{Order-preserving outer-1-string representations}

Now we turn towards positive results and show that
every outer-plane graph has an order-preserving outer-1-string
representation.  We first discuss one existing result that does not quite
achieve this.  It is easy to show that every outer-planar graph can be 
represented as touching-graph of line segments
(see e.g.~\cite{KobourovUeckertVerbeek} for much broader results).
The standard way to do this (see also Fig.~\ref{fig:outplEx}) results,
after extending the segments a bit, in
a segment-representation that is order-preserving and weakly outer-string.  
However, this does not
quite achieve our goal, because the ends of segments are not necessarily
on the outer-face.

\begin{figure}[ht]
\hspace*{\fill}
\includegraphics[page=1,width=0.3\linewidth]{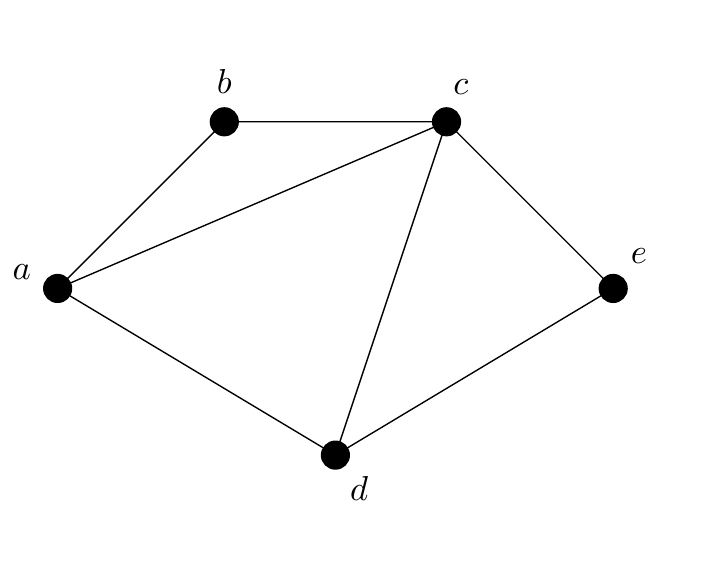}
\hspace*{\fill}
\includegraphics[page=2,width=0.3\linewidth]{figures_outplEx.pdf}
\hspace*{\fill}
\includegraphics[page=3,width=0.3\linewidth]{figures_outplEx.pdf}
\hspace*{\fill}
\caption{An outer-planar graph, a weakly segment-representation 
that is not outer-segment at $\curve{e}$, 
and a representation as circle graph that is not order-preserving at $\curve{c}$.}
\label{fig:outplEx}
\end{figure}

We instead give two other constructions.  The first one uses
that any outer-planar graph is a {\em circle graph}, i.e., the intersection 
graph of chords of a circle \cite{cit:wessel}.  This obviously
gives an end-outer-segment representation, but it need not be
order-preserving (see Fig.~\ref{fig:outplEx}).
Our first construction hence re-proves
this result and maintains invariants to ensure that the representation is
indeed order-preserving.

The resolution in this representation could be very bad, and we therefore
give a second construction where the curves are orthogonal instead.
We use one bend for each vertex curve here, 
and so obtain a $B_1$-VPG-representation.  Since there are $n$ vertices
and at most $n$ bends, the representation
\todo{Minor rewording, and an "at least" changed to "at most".}
can be embedded into a grid of size $O(n) \times O(n)$.  

In our proofs, we use that any 2-connected
outer-planar graph $G$ can be built up as
follows \cite[Lemma 3]{cit:govindran}:  Fix 
an edge $(u,v)$.  Now repeatedly add an \emph{ear}, i.e., a path 
$P = u_0,u_1, \ldots, u_k,u_{k+1}$ with $k\geq 1$ where $(u_0,u_{k+1})$
is an edge on the outer-face of the current graph $G'$, and  
$u_1,\dots,u_k$ are new
vertices that induce a path and have no edges to $G'$ other than $(u_0,u_1)$
and $(u_k,u_{k+1})$.

A crucial requirement of the constructed representation $\calR$ of such
a subgraph is the following {\em order-condition}:
If $w$ and $w'$ are the counterclockwise and clockwise neighbours of $v$ on the outer-face,  then we encounter the neighbours of $v$ in order,
starting with $w$ and ending with $w'$, while walking along
$\curve{v}$.    Put differently, the broken face-vertex-incidence is
the one with the outer-face.  We consider
$\curve{v}$ to be directed so that it intersects first $\curve{w}$ and 
last $\curve{w'}$.

The second crucial ingredient for both proofs is
to reserve for edges (somewhat similar as was done
for faces in \cite{cit:jocg,cit:cccg,cit:chalopin-string,cit:mfcs})  a
region that can be used to attach subgraphs.
%
Thus define a \emph{private region} $\region{u}{v}$ of edge $(u,v)$ to be
a region that 
contains an end of $\curve{u}$ and an end of $\curve{v}$ and does not
intersect any other curve or private regions of $\calR$.
%
Both constructions maintain such a private region $\region{u}{v}$
for every outer-face edge $(u,v)$.  Moreover, if $v$ is the clockwise 
neighbour of $u$, then $S_{\curve{u}\curve{v}}$ 
contains the tail of $\curve{u}$ and the head of $\curve{v}$.

\subsection{Circle-chord representation}

We now re-prove that outer-planar graphs are circle graphs, and show
that furthermore the order can be preserved.

\begin{theorem}
\label{thm:circleGraph}
Every outer-plane graph has an order-preserving representation as intersection graph of chords of a circle $C$.  
\end{theorem}
\begin{proof}
It suffices to prove the claim for a 2-connected outer-planar graph $G$ since
every outer-planar graph $G'$ is an induced subgraph of a 2-connected 
outer-planar graph $G$, and therefore a string representation
for $G$ also yields one for $G'$ by deleting curves of vertices in $G-G'$.

We create a representation $\calR$ while building up the graph via adding
ears, and maintain curve directions and private regions as explained before.
Each private region $\region{u}{v}$ is bounded by parts of circle $C$ and a
chord of $C$ and does not contain the crossing of $\curve{u}$ and $\curve{v}$.  
Further, the tail of $\curve{u}$ 
and the head of $\curve{v}$ are in the interior of the circular arc that bounds
$\region{u}{v}$.

In the base case, $G$ is an edge $(u,v)$ which can be represented by two
chords through the center of $C$.
See Fig.~\ref{fig:circleGraph}.   
\todo{Unimportant change-request: The figure feels different, style-wise,
from the ones you drew with inkscape.  Perhaps redraw.}
We reserve two private regions for $(u,v)$, because the outer-face of
a single-edge graph should be viewed as containing this edge twice (we can add
ears twice at it).  All conditions are easily verified.

\begin{figure}[ht]
\includegraphics[width=0.2\linewidth,page=1]{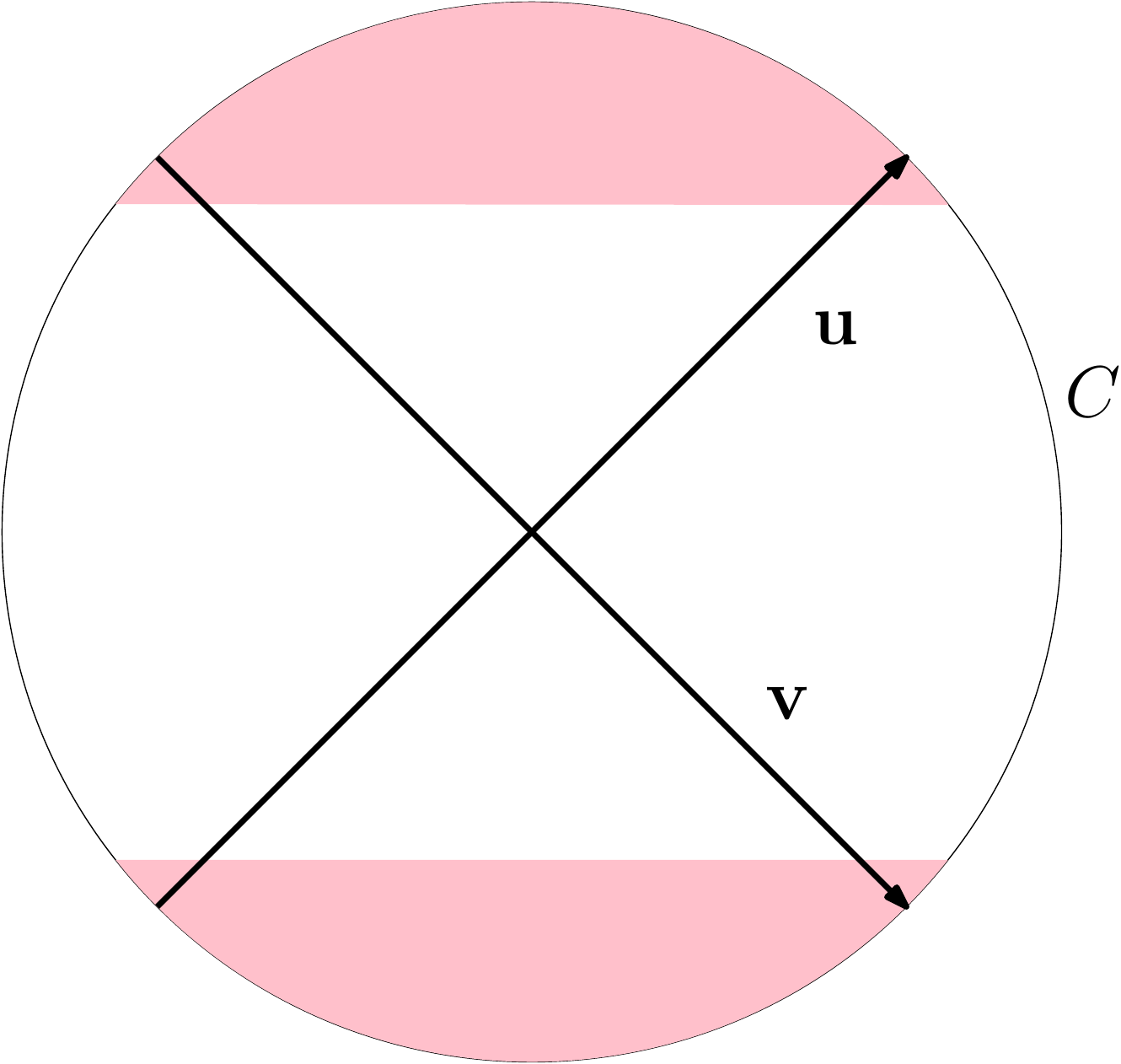}
\hspace*{\fill}
\hspace*{2mm}
\includegraphics[width=0.18\textwidth]{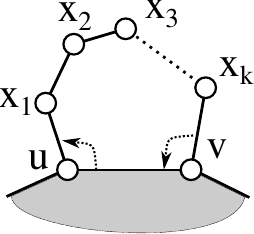}
\hspace*{\fill}
\includegraphics[width=0.56\linewidth,page=2,trim= 0 0 0 200,clip]{figures_chord_base.pdf}
\caption{The base case, and adding chords for an ear.}
\label{fig:circleGraph}
\end{figure}

For the induction step, let us assume that $G$ was obtained by adding an 
ear $P=u,x_1,\dots,x_k,v$ at some edge $(u,v)$, with $u$ the counter-clockwise 
neighbour of $v$ on the outer-face.  Let $C[u,v]$ be the arc of $C$ between
the tail of $\curve{u}$ and the head of $\curve{v}$ that
lies inside $\region{u}{v}$.  Let $u'$ and $v'$ be two points on $C$
just outside $C[u,v]$ but still within $\region{u}{v}$.  If $k=1$, then 
we add $x_1$ by using chord $\overline{u' v'}$ for $\curve{x_1}$.  If $k>1$,
then we insert $2k-2$ points on the interior of $C[u,v]$ and create chords
for $\curve{x_1},\dots,\curve{x_k}$ so that everyone intersects as required.
See Fig.~\ref{fig:circleGraph}, which also shows the private regions that we
define for the new outer-face edges.

Since $\region{u}{v}$ was convex, all new curves are inside it and do not
intersect any other curves.  The orientation of these new curves is determined
by the order-condition: $\curve{x_i}$ should be oriented so that it intersects first
$\curve{x_{i+1}}$ (where $x_0:=u$) and then $\curve{x_{i-1}}$ (where $x_{k+1}:=v$).
In particular this means that the private region $\region{x_i}{x_{i+1}}$ 
contains
the tail of $\curve{x_i}$ and the head of $\curve{x_{i+1}}$, and hence 
satisfies the condition on private regions.

It remains to check that the order-condition is satisfied for $\curve{u}$.
Since $\region{u}{v}$ contained the tail of $\curve{u}$, this means that $\curve{x_1}$
becomes the first curve to be intersected by $\curve{u}$, which is correct
since $x_1$ is the clockwise neighbour of $u$ on the outer-face.  Likewise one
argues that the order-condition holds for $\curve{v}$.   
Hence all conditions hold, and after repeating for all ears we obtain
an order-preserving representation as intersection graph of chords
of a circle.
\qed
\end{proof}

\subsection{$B_1$-VPG representation}

Now we create, for any outer-planar graph,
a $B_1$-VPG representation 
that is order-preserving and outer-string.    However, the ends
will not be on a circle; instead they will lie on a closed curve $S$ 
that we maintain throughout the construction and that surrounds the entire
representation $\calR$ without truly intersecting any curve.  All vertices
are 1-bend poly-lines with slopes $\pm 1$ (after rotating by 45$^\circ$ this
gives the $B_1$-VPG representation); this allows us to use an orthogonal
curve for $S$.
Fig.~\ref{fig:types} illustrates types of private regions that we will use for 
this construction:  $\region{u}{v}$ contains
no bend of $\curve{u}$ or $\curve{v}$, and it is an isosceles right triangle
whose hypotenuse 
lies on $S$.

\begin{figure}
\hspace*{\fill}
\includegraphics[width=0.6\textwidth,trim=0 0 70 0,clip]{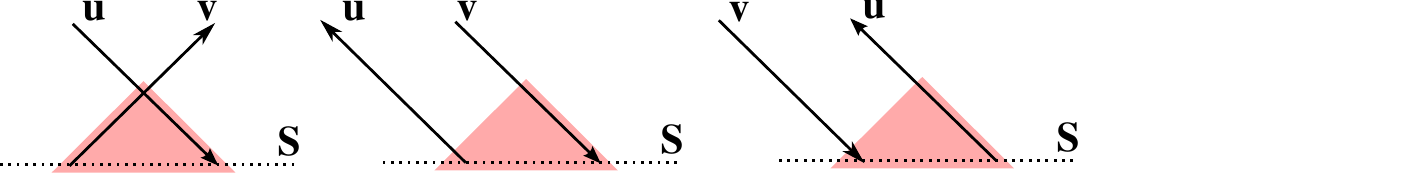}
\hspace*{\fill}
\includegraphics[width=0.3\textwidth]{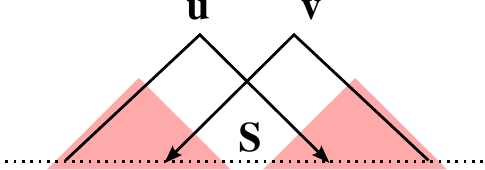}
\hspace*{\fill}
\caption{Three types of private regions (three more can be obtained by
flipping horizontally), and the base case.}
\label{fig:types}
\label{fig:edge}
\end{figure}

\begin{theorem}
\label{thm:rotation}
Every outer-planar graph $G$ has an order-preserving outer-1-string $B_1$-VPG-representation $\calR$.
\end{theorem}
\begin{proof}
As before it suffices to prove the claim for 2-connected outer-planar 
graphs $G$.  We proceed by induction
on the number of vertices,  building $\calR$ while adding ears. 
In the base case, $G$ is 
an edge $(u,v)$ which can be represented by two 1-bend curves 
positioned and oriented as shown in Fig.~\ref{fig:edge}, which also
shows the private region.  We use a horizontal segment for $S$  (this can
be expanded into a closed curve surrounding $\calR$ arbitrarily). 
\todo{Unimportant change-request: Change $S$ into a closed curve}

For the induction step, let us assume that $G$ was obtained by adding an ear $P=u,x_1,\dots,x_k,v$ at some edge $(u,v)$, with $u$ the counter-clockwise neighbour of $v$ on the outer-face.  After possible rotation 
the hypotenuse 
of the private region
$S_{\curve{u}\curve{v}}$ is horizontal with $\region{u}{v}$ above it. 
We distinguish cases: 
\begin{enumerate}
\item 
\textit{$\curve{u}$ and $\curve{v}$ have different slopes 
in $\region{u}{v}$ and $k=1$ (i.e. we add one
vertex $x$).}
\todo{Unimportant change-request for the representations:
Make the parts of $\curve{u}$ and $\curve{v}$ that get cut dotted.}

\begin{figure}
\hspace*{\fill}
\includegraphics[width=0.18\textwidth]{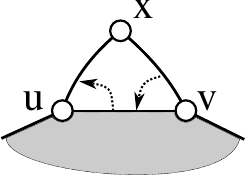}
\hspace*{\fill}
\includegraphics[width=0.38\textwidth]{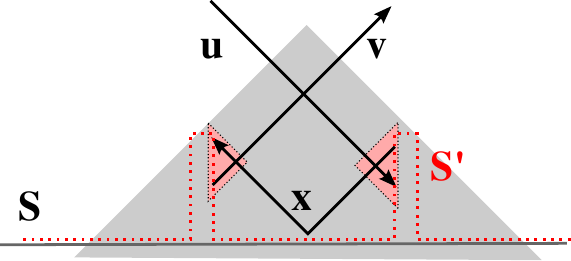}
\hspace*{\fill}
\includegraphics[width=0.38\textwidth]{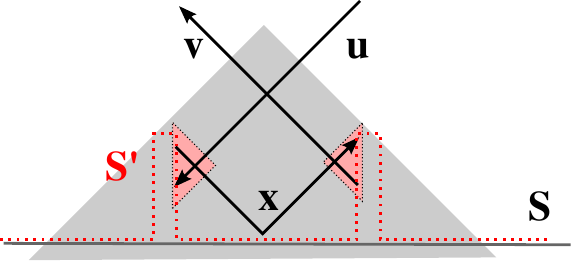}
\hspace*{\fill}
\caption{Adding a single node if $\curve{u}$ and $\curve{v}$ have different slopes.} 
\label{fig:case1}
\end{figure}

We add a 1-bend curve $\curve{x}$ with the bend pointing downwards.
See Fig.~\ref{fig:case1}, which also shows the private regions that we
define for $(u,x)$ and $(x,v)$.
Curve $\curve{x}$ fits entirely inside $\region{u}{v}$ by placing the bend in the interior of 
$\region{u}{v}$ and shortening $\curve{u}$ and $\curve{v}$ appropriately so that the ends of $\curve{x}$ are
vertically aligned with those of $\curve{u}$ and $\curve{v}$. 
We can now easily find a new curve $S'$ by adding ``detours''
to $S$ that reach the hypotenuses of the new private regions. 
These detours are inside $\region{u}{v}$ and hence intersect no other
curves (since we shortened $\curve{u}$ and $\curve{v}$). So the 
new curve $S'$ is a closed curve that surround the new representation
as desired.

The orientation of $\curve{x}$ is again determined by the order-condition,
and exactly as in Theorem~\ref{thm:circleGraph} one argues that this respects
the order-condition at $\curve{u}$ and $\curve{v}$, since our choice
of curve for $\curve{x}$ ensures that it crosses $\curve{u}$ {\em after}
the crossing of $\curve{u}$ with $\curve{v}$.

\item \textit{$\curve{u}$ and $\curve{v}$ have different slopes in
$\region{u}{v}$ and $k>1$ (i.e. we add at least two vertices $x_1,\dots,x_k$.)}

\begin{figure}
\hspace*{\fill}
\includegraphics[width=0.18\textwidth]{figures_fig5_graph.pdf}
\hspace*{\fill}
\includegraphics[width=0.38\textwidth]{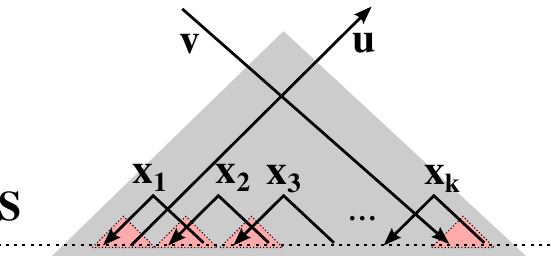}
\hspace*{\fill}
\includegraphics[width=0.38\textwidth]{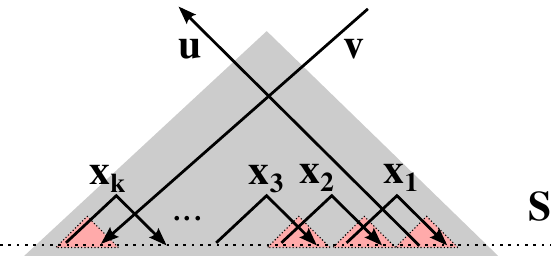}
\hspace*{\fill}
\caption{Adding 2 or more nodes if $\curve{u}$ and $\curve{v}$ have different slopes. }
\label{fig:case2}
\end{figure}

We add a path of 1-bend curves $\curve{x_1}, \curve{x_2}, \ldots, \curve{x_k}$ with their bends at the top, and define private regions as illustrated in Fig.~\ref{fig:case2}.  Each curve $\curve{x_i}$ is oriented as required by the order-condition, and again one verifies the order-condition for $\curve{u}$ and $\curve{v}$.  We can re-use the same $S$.
 
\item \textit{$\curve{u}$ and $\curve{v}$ have the same slope inside
$\region{u}{v}$.}

\begin{figure}
\hspace*{\fill}
\includegraphics[width=0.18\textwidth]{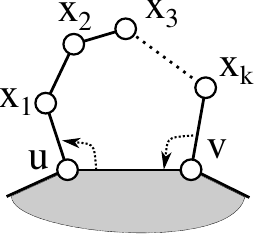}
\hspace*{\fill}
\includegraphics[width=0.38\textwidth]{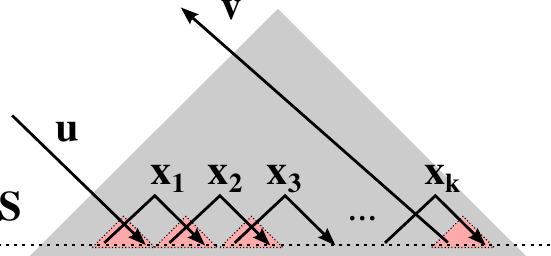}
\hspace*{\fill}
\includegraphics[width=0.38\textwidth]{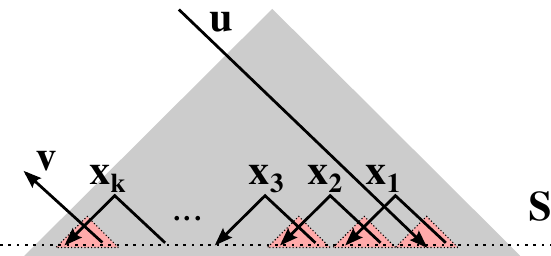}
\hspace*{\fill}
\caption{Adding one or more vertices if $\curve{u}$ and $\curve{v}$ have the same slope.  We only show two of the four possible configurations.}
\label{fig:case6}
\label{fig:case3}
\end{figure}
We add a path of 1-bend curves $\curve{x_1}, \curve{x_2}, \ldots, \curve{x_k}$ (possibly $k=1$) with their bends at the top, and define private regions as illustrated in Fig.~\ref{fig:case3}.  Each curve $\curve{x_i}$ is oriented as required by the order-condition, and one verifies all conditions using the same $S$.

\end{enumerate}
After having represented the entire graph in this way, we are 
order-preserving due to the order-condition, outer-string due
to poly-line $S$, and $B_1$-VPG (after a 45$^\circ$-rotation) since every
curve has one bend.
\qed
\end{proof}

In our $B_1$-VPG-representation, every vertex-curve is an \L
in one of the four possible rotations 
\lefthalfcup, \righthalfcup, \lefthalfcap, \righthalfcap.
(All four may be used, since private regions get rotated in Case 1.)
We would have preferred a representation that uses \L (or
the two shapes \lefthalfcup and \righthalfcup), because then  the
stretching-techniques 
by Middendorf and Pfeiffer \cite{cit:stretching} could have been applied
to obtain another segment-representation.  
It is easy to create a representations with \L only if we need
not be order-preserving (use \rotatebox{45}{\righthalfcap} in Case 1) or need
not be outer-string (see also Lemma~\ref{lem:SP}), but finding an
outer-string order-preserving representation using only {\L}s
remains open.

\section{Beyond outer-planar graphs?}

One wonders what other graph classes might have order-preserving
1-string representations, preferably outer-string ones.  We
study this here for some graph classes.
%
%
We start with the {\em series-parallel graphs}, which are the same as the
partial 2-trees, and hence
generalize outer-planar graphs.  

\begin{lemma}
\label{lem:SP}
Every series-parallel graph $G$ has a 1-string representation with {\L}s
that is order-preserving for some planar embedding of $G$.
\end{lemma}
\begin{proof}
It is easy to show that every 2-tree has a representation by touching
true {\L}s,   i.e., each vertex is assigned an {\L} (not  rotated
and not degenerated into a line segment),
curves are disjoint except at ends, and $(u,v)$ is an
edge if and only if the end of $\curve{u}$ lies on the interior of $\curve{v}$ 
or vice versa.%
\footnote{We have not been able to find a direct reference for this, but it follows
for example from the works of Chaplick et al. 
\cite{ChaplickKobourovUeckert2012} or with an iterative approach similar to 
the 6-sided contact representations in \cite{AlamBiedl2011}. } 
See also Fig.~\ref{fig:SP}. 
Extending
the {\L}s slightly gives a 1-string representation, and it is order-preserving
for a planar embedding easily derived from the touching {\L} representation.
\todo{minor rewordings}
Details are provided in Appendix~\ref{sec:appendix}.
\qed
\end{proof}

\begin{figure}[t]
\hspace*{\fill}
\includegraphics[width=0.22\linewidth]{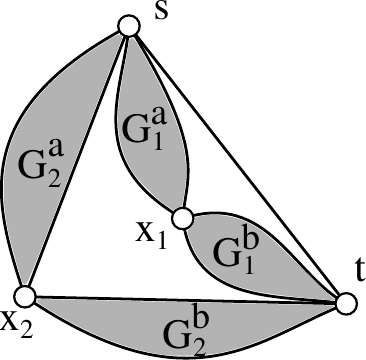}
\hspace*{\fill}
\includegraphics[width=0.25\linewidth]{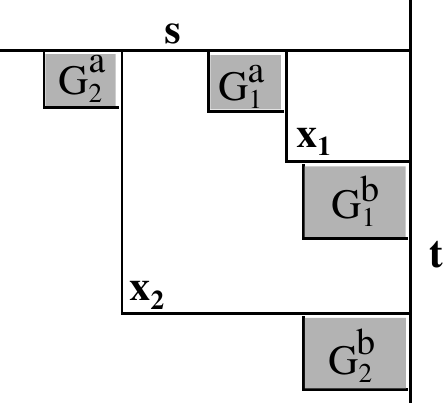}
\hspace*{\fill}
\includegraphics[width=0.25\linewidth]{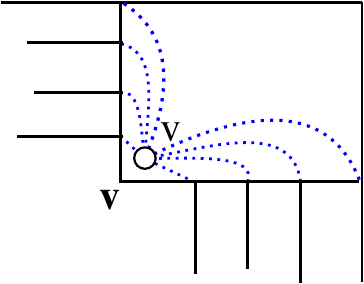}
\hspace*{\fill}
\caption{Representing series-parallel graphs by touching {\protect\L}s, and converting
this into a planar drawing with the same order.}
\label{fig:SP}
\end{figure}

It would be interesting to know whether this result can be extended to the
so-called {\em planar Laman-graphs}, which have a representation by
touching {\L}s \cite{KobourovUeckertVerbeek}, but not all {\L}s
are necessarily in the same rotation and so it is not clear whether
this is order-preserving. Of particular interest would be planar bipartite
graphs, which can even be represented by horizontal and vertical touching
line segments \cite{FMP91}, but again it is not clear how to make this
order-preserving.

As for having strings additionally end at the contour for series-parallel
graphs: this is not always possible.
Let $H$ be the graph obtained by subdividing every edge in a $K_{2,3}$; one 
verifies that $H$ is series-parallel.  It is easy to see (see also
\cite{Cabello2016}) that $H$ cannot be outer-string, since $K_{2,3}$ is not
outer-planar.  So $H$ has no outer-string representation, much less one
that is 1-string and order-preserving.

Now we turn to partial 3-trees.
We showed in Theorem~\ref{thm:notOrder} that there exist planar 3-trees
(hence partial 3-trees) that do not have an order-preserving 1-string
representation.  We now study some subclasses of partial 3-trees
that are superclasses of outer-planar graphs.

An {\em IO-graph} is a planar graph $G$ that has an independent set $I$ such
that $G-I$ is a 2-connected outer-planar graph $O$ for which all vertices
in $I$ are inside inner faces of $O$.  A {\em Halin}-graph is a graph
that consists of a tree $T$ and a cycle $C$ that connects all leaves of $T$.
Both types of graphs are well-known to be partial 3-trees.
In \cite{cit:cccg}, we gave 1-string 
representations for both Halin graphs and IO-graphs; the latter uses only 
unrotated {\L}s.
Independently, Francis and Lahiri also constructed 1-string representations
of Halin-graphs, using only unrotated {\L}s \cite{cit:francis}.
Inspection of both constructions shows that 
these respect the standard planar embedding (where $O$ respectively $C$ 
is one face).  We hence have:

\begin{theorem}[based on \cite{cit:cccg,cit:francis}]
Every IO-graph and every Halin-graph has an order-preserving 1-string
representation in which every vertex is an {\L}.
\end{theorem}

In these constructions, the ends of the strings are not on the outer-face,
and we now show that this is unavoidable.  This is obvious for Halin-graphs,
since the subdivided $K_{2,3}$ is an induced subgraph of a Halin-graph.
As for IO-graphs, define the {\em wheel} $W_n$
be the graph that consists of a cycle $C=\{v_1,\dots,v_n\}$ with $n$ vertices 
and one universal vertex $c$ connected to all of them.  Let the {\em
extended wheel-graph} $W^+_n$ be the wheel-graph $W_n$ with additionally
a vertex $w_i$ incident to $v_i$ and $v_{i+1}$ for $i=1,\dots,n$ (and
$w_{n+1}:=w_1$).  
Notice that $W_n^+$ is an IO-graph. The proof of the following is presented in Appendix~\ref{sec:appendix}.

\begin{theorem}
\label{thm:IOnotEndOrder}
For $n\geq 7$, the IO-graph $W^+_n$ has no order-preserving outer-1-string 
representation.
\end{theorem}
%

\section{Final remarks}
\label{sec:conclusions}

In this paper, we studied 1-string representations 
that respect
a planar embedding.
As for open problems,
	what other graph classes have order-preserving 1-string representations?
	A natural candidate to investigate would be the 2-outer-planar graphs,
	for which Lemma~\ref{lem:stellation} cannot be applied since a 
	triple-stellation is never 2-outer-planar.  
Other interesting candidates would be planar bipartite graphs
(or more generally planar Laman-graphs), or planar 4-connected graphs.

	Secondly, what is the complexity of testing whether an order-preserving
	1-string representation exists?  Given the NP-hardness of  the
	abstract graph realization problem 
	\cite{cit:kratochvil-II,cit:middendorf}, this is very likely NP-hard
	if we are allowed to prescribe an arbitrary rotation scheme (not
	from a planar drawing).  But is it NP-hard for plane graphs?

\todo{fairly major rewording here}
One unsatisfactory aspect of our definition of ``order-preserving'' is
that graphs with an end-contact representation (i.e., with disjoint strings 
where for every edge one string ends on the other string) do not automatically
have an order-preserving 1-string representation: We can obtain a 1-string
representation by extending the strings slightly, 
but it does not need to be order-preserving.
A reviewer hence suggested to us the following alternate model:  Thicken each
string slightly, and
consider the cyclic order of intersections while walking around the
thickened string.
Let now ``order-preserving'' mean that the cyclic order 
of neighbours around a vertex forms a subsequence of the intersections encountered while
walking ``around'' its string. 
With this, any end-contact representation becomes an order-preserving
1-string representation after extending the curves a bit.  
This includes for example planar bipartite graphs and Laman graphs.
Since this model's restriction is weaker, all our positive results transfer,
but the proofs of the negative results no longer hold.  
Are there plane graphs that do not have an order-preserving 1-string 
representation in this new model? 

\bibliographystyle{splncs03}
\bibliography{string.bib}

\newpage
\appendix
\section{Appendix}
\label{sec:appendix}

\todo{Removed a few repeated things.}
\begin{proof}[of Lemma~\ref{lem:SP}]
Start with the representation by touching true {\L}s explained in the
main part of the paper.
%
%
For any contact-representation with true {\L}s (neither
rotated nor degenerated into a horizontal or vertical line segment), we can
create a planar drawing
that matches the order of touching-points along each {\L}.  
Namely, draw a point for $v$
slightly above and to the right of the corner of the bend in $\curve{v}$.
Connect $v$ to all touching-points on $\curve{v}$, and to the two ends of
$\curve{v}$.    
Because every curve is an {\L}, the curves whose ends touch
$\curve{v}$ all come from the left at the vertical segment of $\curve{v}$
or from the bottom at the horizontal segment of $\curve{v}$.  Therefore
the added lines do not cross any curves and so give a planar drawing of $G$
that is clearly respected by the representation.
Extending
the {\L}s slightly hence gives the desired 1-string representation.
\qed
\end{proof}

\begin{proof}[of Theorem~\ref{thm:IOnotEndOrder}]
Assume for contradiction that it did, and consider the induced representation
$\calR_W$ of $W_n$.  Let the naming of cycle $C$ be such
that $\curve{c}$ intersects $\curve{v_1},\dots,\curve{v_n}$ in this order.
Define as before $\curve{u}[v,w]$ (for any 2-path $v,u,w$) to be the stretch 
of $\curve{u}$ between the intersection with $\curve{v}$ and $\curve{w}$.
Now define $R$ to be the region bounded by $\curve{c}[v_1,v_n]$ (which is
almost the entire curve $\curve{c}$), as well as $\curve{v_n}[c,v_1]$ and
$\curve{v_1}[v_n,c]$ (which exist since $(v_1,v_n)$ is an edge). 
See also Fig.~\ref{fig:IOnotEndOrder}.

Consider $v_i$ for $i=3,4,5$, which is adjacent to neither $v_1$ nor $v_n$.
$\curve{v_i}$ intersects the boundary of $R$ (because it intersects
$\curve{c}[v_1,v_n]$ by assumption), but does not intersect it twice, else
it would intersect $\curve{c}$ twice or intersect $\curve{v_1}$ or 
$\curve{v_n}$.  Hence one end of $\curve{v_i}$ is inside $R$ while the
other one is outside, and so not both ends of $\curve{v_i}$ can be on
the contour for $i=3,4,5$. 

This shows that $W_n$ is not outer-1-string in the sense that for some
vertex not both ends of the curves are on the contour.  
Now consider $W_n^+$, and the vertices $w_3$ and $w_4$ that were added
at $v_4$  when creating $W_n^+$.  
Since $w_3$ and $w_4$ are
adjacent to none of $c,v_1,v_n$, and since the drawing is outer-string,
both $\curve{w_3}$ and $\curve{w_4}$ 
(and therefore their intersections with $\curve{v_4}$) must
be outside $R$.

So walking along $\curve{v_4}$ starting at the end inside $R$, we
encounter $\curve{c}$ and then one of $\{\curve{w_3},\curve{w_4}\}$.
We assume that we encounter $\curve{w_3}$ before $\curve{w_4}$;
the other case is symmetric (and results in $\curve{v_5}$
having no end on the contour).  Consider the region $R'$ enclosed by
$\curve{v_4}[c,w_4]$, $\curve{w_4}[v_4,v_5]$, $\curve{v_5}[w_4,c]$ and
$\curve{c}[v_5,v_4]$. Since $\curve{w_4}$ is outside $R$, so is $R'$.
Curve $\curve{v_3}$ intersects $\delta R'$, because
it intersects $\curve{v_4}$, and this intersection must be on $\curve{v_4}[c,w_3]$
to preserve the order of edges around $v_4$ (and since we know that
$\curve{c},\curve{w_3},\curve{w_4}$ intersect $\curve{v_4}$ in this order).
Curve $\curve{v_3}$ cannot intersect $\delta R'$ again, else it would intersect
$\curve{c}$ or $\curve{v_4}$ twice or would intersect $\curve{w_4}$ or $\curve{v_5}$,
which it shouldn't.  Therefore one end of $\curve{v_3}$ is inside $R'$, which is outside $R$.
The other end of $\curve{v_3}$ is inside $R$.  So neither end of $\curve{v_3}$ is on
the contour.  Contradiction.
\qed
\end{proof}

\begin{figure}[b]
\hspace*{\fill}
\includegraphics[width=0.3\linewidth]{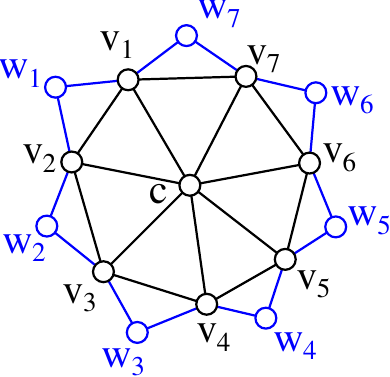}
\hspace*{\fill}
\includegraphics[width=0.3\linewidth]{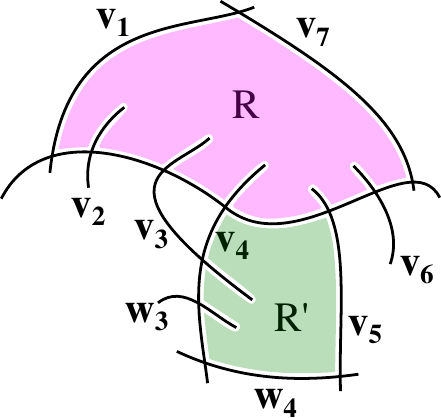}
\hspace*{\fill}
\caption{For the proof of Theorem~\ref{thm:IOnotEndOrder}.}
\label{fig:IOnotEndOrder}
\label{fig:IO}
\end{figure}

\end{document}